\documentclass{article}

\usepackage{microtype}
\usepackage{graphicx}
\usepackage{subcaption}
\usepackage{booktabs} 
\usepackage{multirow} 
\usepackage{enumitem} 

\usepackage{booktabs}
\usepackage{pifont}
\usepackage{amsmath} 
\usepackage{multirow}

\usepackage{hyperref}

\usepackage[preprint]{icml2026}

\usepackage{amsmath}
\usepackage{amssymb}
\usepackage{mathtools}
\usepackage{amsthm}
\newtheorem*{proposition*}{Proposition} 

\usepackage{tikz}
\usetikzlibrary{shapes.geometric,arrows.meta,positioning,fit,calc,backgrounds}

\usepackage[capitalize,noabbrev]{cleveref}

\theoremstyle{plain}
\newtheorem{theorem}{Theorem}[section]
\newtheorem{proposition}[theorem]{Proposition}

\theoremstyle{definition}

\theoremstyle{remark}

\newcommand{\cmark}{\ding{51}}
\newcommand{\xmark}{\ding{55}}

\usepackage[textsize=tiny]{todonotes}

\icmltitlerunning{DSFlow: Dual Supervision and Step-Aware Architecture}

\begin{document}

\twocolumn[
  \icmltitle{DSFlow: Dual Supervision and Step-Aware Architecture \\ for One-Step Flow Matching Speech Synthesis}

  \icmlsetsymbol{equal}{*}

  \begin{icmlauthorlist}
    \icmlauthor{Bin Lin}{StepFun}
    \icmlauthor{Peng Yang}{StepFun}
    \icmlauthor{Chao Yan}{StepFun}
    \icmlauthor{Xiaochen Liu}{StepFun}
    \icmlauthor{Wei Wang}{StepFun}
    \icmlauthor{Boyong Wu}{StepFun}
    \icmlauthor{Pengfei Tan}{StepFun}
    \icmlauthor{Xuerui Yang}{StepFun}
  \end{icmlauthorlist}

  \icmlaffiliation{StepFun}{StepFun}

  \icmlkeywords{Flow Matching, Knowledge Distillation, Text-to-Speech, Generative Models, Classifier-Free Guidance}

  \vskip 0.3in
]



\printAffiliationsAndNotice{}  

\begin{abstract}
Flow-matching models have enabled high-quality text-to-speech synthesis, but their iterative sampling process during inference incurs substantial computational cost. Although distillation is widely used to reduce the number of inference steps, existing methods often suffer from process variance due to endpoint error accumulation. Moreover, directly reusing continuous-time architectures for discrete, fixed-step generation introduces structural parameter inefficiencies.
To address these challenges, we introduce DSFlow, a modular distillation framework for few-step and one-step synthesis. DSFlow reformulates generation as a discrete prediction task and explicitly adapts the student model to the target inference regime. It improves training stability through a dual supervision strategy that combines endpoint matching with deterministic mean-velocity alignment, enforcing consistent generation trajectories across inference steps. In addition, DSFlow improves parameter efficiency by replacing continuous-time timestep conditioning with lightweight step-aware tokens, aligning model capacity with the significantly reduced timestep space of the discrete task.
Extensive experiments across diverse flow-based text-to-speech architectures demonstrate that DSFlow consistently outperforms standard distillation approaches, achieving strong few-step and one-step synthesis quality while reducing model parameters and inference cost.
\end{abstract}

\section{Introduction}

Flow Matching models~\cite{lipman2023flow,liu2023flow,gat2024discrete,liu2025flow, guo2025splitmeanflow} have emerged as a powerful framework for text-to-speech synthesis, achieving high-quality generation through iterative refinement along learned probability flows. Similar to their success in image~\cite{esser2024scaling, Flow-albergo2023, RF-liu2022} and video generation~\cite{bar2024lumiere,kong2024hunyuanvideo,VideoFlow-Jin2024}, these models learn to map data distributions through continuous flows. However, flow-based models face a fundamental trade-off between sample quality and inference efficiency: their iterative sampling process requires tens to hundreds of neural function evaluations (NFEs), leading to high latency that prohibits real-time deployment in applications such as voice assistants and interactive dialogue systems.

To address this efficiency bottleneck, recent work has explored techniques to reduce NFEs while maintaining generation quality. Consistency models~\cite{salimans2022progressive,song2023consistency,facm-peng2025, scm, liu2024audiolcm} achieve few-step generation by learning to map any point along a sampling trajectory directly to its endpoint. During distillation, the student is supervised to match the teacher's final output regardless of the starting timestep. However, this endpoint-only supervision provides sparse learning signals: gradient information must propagate backward through accumulated errors across all sampling steps, leading to training instability. An alternative approach, MeanFlow~\cite{geng2025mean}, models the average velocity field rather than instantaneous velocity, enabling one-step generation trained from scratch without distillation. While MeanFlow has shown strong results in image generation, its training process requires Jacobian-vector products (JVPs)---a computationally expensive operation that significantly increases memory consumption and is incompatible with customized CUDA operators. 

In this work, we propose DSFlow, a distillation framework that addresses these limitations through improved supervision and architectural adaptation. We introduce dual supervision, combining endpoint matching with velocity field alignment to provide complementary learning signals that enhance training stability. Drawing from MeanFlow's average velocity formulation, we develop a JVP-free implementation that retains dense trajectory supervision while eliminating computational overhead. Finally, we propose step-aware tokens: rather than inheriting continuous-time embeddings, we employ learnable discrete representations corresponding to the fixed sampling steps, achieving substantial parameter efficiency and enabling step-specific specialization. Our contributions can be summarized as:

\begin{itemize}[leftmargin=*, itemsep=1pt, parsep=0pt, topsep=1pt]
    \item We propose dual supervision that combines endpoint matching with velocity field alignment, providing complementary learning signals that improve training stability and convergence.
    \item We develop a velocity matching formulation inspired by MeanFlow's average velocity, achieving dense trajectory supervision without the computational overhead of Jacobian computation.
    \item We introduce step-aware tokens that replace continuous-time embeddings with learnable discrete representations, substantially reducing parameters in the time-conditioning pathway while enabling step-specific specialization.
    \item We demonstrate through comprehensive experiments that DSFlow achieves teacher-level quality within a single inference step on text-to-speech benchmarks, enabling real-time synthesis with significantly reduced computational costs compared to existing methods.
\end{itemize}

\section{Related Work}
\label{sec:related}

\paragraph{Flow Matching and Efficient Inference.}
Flow Matching~\cite{lipman2023flow, liu2023flow} has emerged as an efficient alternative to diffusion-based generative modeling by enabling simulation-free training of continuous normalizing flows. Owing to its stability and strong empirical performance, flow matching has been successfully adopted in conditional generation tasks, including text-to-speech synthesis~\cite{mehta2024matcha, du2024cosyvoice}. Despite these advantages, high-quality synthesis with flow matching typically relies on numerical integration with multiple function evaluations at inference time. As the number of integration steps directly impacts latency and computational cost, especially in real-time or large-scale deployment scenarios, improving inference efficiency remains a central challenge. This has motivated growing interest in distillation and acceleration techniques that aim to reduce the number of integration steps while preserving synthesis quality.

\paragraph{Distillation Methods for Flow-Based Models.}
A widely explored strategy for accelerating flow-based models is endpoint distillation, where a student model is trained to reproduce the final output of a multi-step teacher using significantly fewer integration steps. Progressive distillation~\cite{salimans2022progressive} iteratively reduces the number of steps through endpoint supervision and has shown promising results in diffusion and flow-based models. However, because supervision is applied only at the trajectory endpoint, small prediction errors can accumulate across integration steps, often leading to high training variance and unstable convergence. To mitigate this issue, MeanFlow~\cite{geng2025mean} proposes supervising the student by matching the teacher’s mean velocity field, thereby providing denser, process-level supervision that reduces variance. This approach, however, relies on Jacobian-vector product computation, which substantially increases training cost. IntMeanFlow~\cite{intmeanflow2025} avoids explicit Jacobian computation by approximating velocity integrals, but may still suffer from endpoint drift when distilled to very few steps. Collectively, existing distillation methods tend to favor either endpoint accuracy or trajectory-level supervision, introducing trade-offs among training stability, computational overhead, and final synthesis quality.

\paragraph{Controllability in Distilled Generative Flow-Matching Models.}
Classifier-free guidance (CFG)~\cite{ho2022cfg} is a standard mechanism for improving sample quality and controllability in diffusion and flow-based generative models by interpolating between conditional and unconditional predictions:
\begin{equation}
v_{\text{cfg}}(x_t, t, c, w)
=
(1-w)\, v_\theta(x_t, t, \emptyset)
+
w\, v_\theta(x_t, t, c),
\label{eq:cfg}
\end{equation}
where the guidance scale $w$ controls the trade-off between fidelity and diversity. Conventional CFG training explicitly allocates model capacity to both conditional and unconditional branches, typically by dropping conditioning information for a subset of training samples. In distillation settings, however, teacher targets are often generated with guidance enabled, leading the student to implicitly absorb the effect of CFG during training. As a result, prior work commonly assumes that inference-time CFG is unnecessary or ineffective for distilled models. Recent observations suggest that limited controllability may still be beneficial if the unconditional branch remains well-formed. Our work builds on this insight by introducing a lightweight regularization strategy that preserves unconditional predictions during distillation, enabling weak inference-time CFG for quality adjustment without compromising student–teacher alignment.

\paragraph{Architectural Design for Distilled Models.}
Most flow-matching architectures adopt modulation-based conditioning mechanisms, such as adaLN-Zero~\cite{peebles2023dit}, which inject conditioning signals into each Transformer layer and are designed to model continuous-time dynamics with high expressive capacity. While effective for continuous-time flow modeling, such designs may be mismatched to distilled settings where inference is restricted to a small, discrete set of steps. In contrast, token-based conditioning mechanisms, which encode conditioning information directly as learnable tokens, have proven effective in multimodal and unified architectures~\cite{radford2021clip, bao2023uvit}. These approaches offer a more compact and flexible alternative by leveraging self-attention to integrate conditioning signals. Despite their success in other domains, architectural adaptation for discrete-step distilled flow models has received limited attention. Our work addresses this gap by introducing step-aware tokenization, which aligns architectural capacity with the reduced complexity of few-step inference.

\section{Method}
\label{sec:method}

\begin{figure*}[h]
\centering
\includegraphics[width= \textwidth]{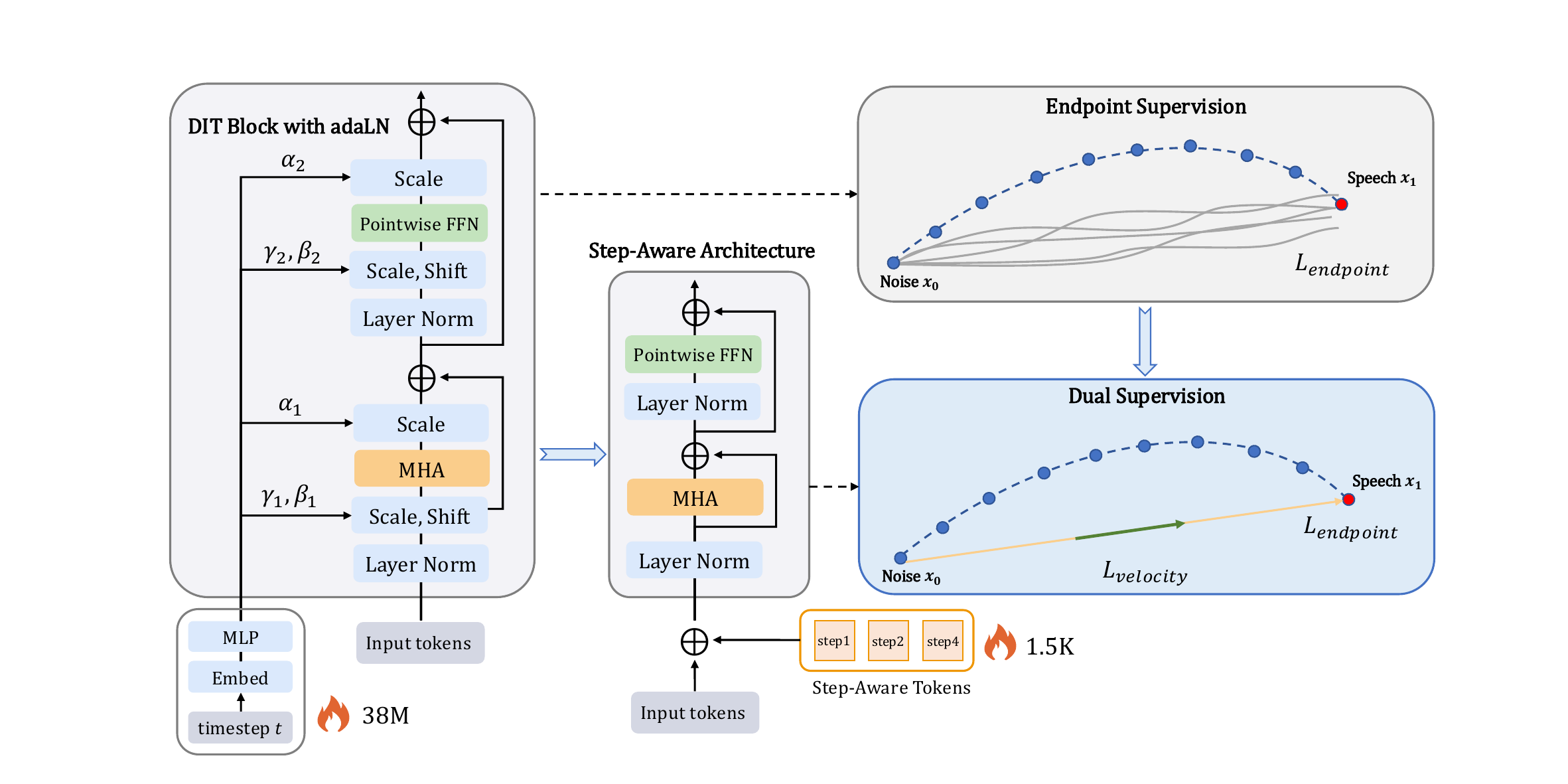}
\caption{Overview of the DSFlow framework. The left and center panels show the transition from a DiT block with adaLN-Zero conditioning to the proposed step-aware architecture, where the heavy time-modulation network is replaced by lightweight step-aware tokens. The right panel illustrates dual supervision, which combines endpoint matching ($\mathcal{L}_{\mathrm{endpoint}}$) with deterministic mean velocity alignment ($\mathcal{L}_{\mathrm{velocity}}$) to guide the student along the teacher’s mean trajectory (green vectors), improving process consistency over endpoint-only distillation without additional Jacobian computation.}
\label{fig:overview}
\end{figure*}

\subsection{Problem Formulation}

We adopt the flow matching framework~\cite{lipman2023flow,diffusion,song2019score,ho2020denoising,scoresde,nichol2021improved,rombach2021high}, which learns a time-dependent velocity field $v_\theta : \mathbb{R}^D \times [0,1] \times \mathcal{C} \rightarrow \mathbb{R}^D$ parameterized by $\theta$. The velocity field defines the ordinary differential equation:
\begin{equation}
\frac{d x_t}{d t} = v_\theta(x_t, t, c), \quad t \in [0,1],
\end{equation}
where $x_t \in \mathbb{R}^D$ denotes the intermediate state along the flow trajectory, $t$ is the continuous flow time, and $c \in \mathcal{C}$ represents conditioning information (e.g., text or prompt features).

\paragraph{Endpoint Distillation.}
A commonly used distillation objective trains a student model to match the final output of a pretrained teacher. Formally, the endpoint distillation loss is defined as
\begin{equation}
\begin{aligned}
\mathcal{L}_{\mathrm{endpoint}}
&=
\mathbb{E}_{x_0, c}
\Big[
\big\|
\mathrm{ODE}_K(v_{\mathcal{S}}, x_0, c)
\\
&\qquad\qquad
-
\mathrm{ODE}_N(v_{\mathcal{T}}, x_0, c)
\big\|^2
\Big],
\end{aligned}
\label{eq:endpoint_loss}
\end{equation}
where $x_0 \in \mathbb{R}^D$ is a sample from the target data distribution, $v_{\mathcal{T}}$ and $v_{\mathcal{S}}$ denote the teacher and student velocity fields, and $\mathrm{ODE}_N(\cdot)$ and $\mathrm{ODE}_K(\cdot)$ denote numerical ODE solvers with $N$ and $K$ integration steps, respectively. This formulation does not impose constraints on intermediate states and may lead to high-variance training due to error accumulation.

In this work, the teacher velocity field $v_{\mathcal{T}}$ is trained with $N = 10$ inference steps, while the student velocity field $v_{\mathcal{S}}$ is trained to operate with $K \in \{1, 2, 4\}$ steps. The corresponding model parameters are denoted by $\theta_{\mathcal{T}}$ and $\theta_{\mathcal{S}}$. Step counts are chosen as powers of two to facilitate progressive distillation and stable time discretization. Both teacher and student models take $(x_t, t, c)$ as input. The objective is to learn a compact student model that achieves perceptual quality comparable to the teacher while substantially reducing inference cost, characterized by $K \ll N$ and $|\theta_{\mathcal{S}}| < |\theta_{\mathcal{T}}|$.

\subsection{Dual Supervision for Process Alignment}

Conventional flow distillation typically relies on endpoint matching (Eq.~\ref{eq:endpoint_loss}), which supervises only the final state of the integration trajectory. This approach can exhibit compromised training stability, as small prediction errors may accumulate during ODE integration. Alternatively, velocity field supervision~\cite{geng2025mean} provides denser trajectory-level constraints but requires Jacobian-vector product (JVP) computation and may introduce endpoint bias when used in isolation. We therefore combine endpoint matching and velocity alignment to leverage their complementary strengths: endpoint supervision anchors the final solution, while velocity supervision provides intermediate process guidance.

\paragraph{Mean Velocity Computation.}
For each integration interval $[t_{\mathrm{start}}, t_{\mathrm{end}}]$, we compute a deterministic estimate of the teacher’s mean velocity:
\begin{equation}
\bar{v}_{\mathcal{T}}(x_{t_{\mathrm{mf}}}, t_{\mathrm{mf}}, c)
=
\frac{
x_{\mathcal{T}}(t_{\mathrm{end}})
-
x_{\mathcal{T}}(t_{\mathrm{start}})
}{
t_{\mathrm{end}} - t_{\mathrm{start}}
},
\label{eq:mean_velocity}
\end{equation}
where $t_{\mathrm{mf}} = (t_{\mathrm{start}} + t_{\mathrm{end}})/2$ denotes the midpoint of the interval. The states $x_{\mathcal{T}}(t_{\mathrm{start}})$ and $x_{\mathcal{T}}(t_{\mathrm{end}})$ are obtained via the teacher’s ODE solver. The intermediate evaluation point is approximated by linear interpolation:
\begin{equation}
x_{t_{\mathrm{mf}}}
=
(1 - t_{\mathrm{mf}})\, x_{\mathcal{T}}(t_{\mathrm{start}})
+
t_{\mathrm{mf}}\, x_{\mathcal{T}}(t_{\mathrm{end}}).
\end{equation}
This construction yields a stable velocity target while avoiding the additional computational overhead associated with JVP-based methods~\cite{geng2025mean}.

\paragraph{Dual Supervision Loss.}
We define the dual supervision objective as a weighted combination of endpoint matching and velocity alignment:
\begin{equation}
\begin{aligned}
\mathcal{L}_{\mathrm{dual}}
&=
\alpha \,
\mathbb{E}\!\left[
\left\|
x_{\mathcal{S}}(t_{\mathrm{end}})
-
x_{\mathcal{T}}(t_{\mathrm{end}})
\right\|^2
\right]
\\
&\quad +
(1-\alpha)\,
\mathbb{E}\!\left[
\left\|
v_{\mathcal{S}}(x_{t_{\mathrm{mf}}}, t_{\mathrm{mf}}, c)
-
\bar{v}_{\mathcal{T}}
\right\|^2
\right],
\end{aligned}
\label{eq:dual_loss}
\end{equation}
where $\alpha$ controls the relative contribution of the two terms. For $K$-step distillation, the loss is averaged across all subintervals:
\begin{equation}
\mathcal{L}_{\mathrm{Stage1}}
=
\frac{1}{K}
\sum_{k=1}^{K}
\left[
\alpha\, \mathcal{L}^{(k)}_{\mathrm{endpoint}}
+
(1-\alpha)\, \mathcal{L}^{(k)}_{\mathrm{velocity}}
\right].
\label{eq:stage1_full_loss}
\end{equation}
We set $\alpha = 0.7$ to emphasize endpoint accuracy while retaining substantial intermediate supervision. Empirical stability improvements are analyzed in Section~\ref{sec:ablation}.

\subsection{Step-Aware Tokenization for Structural Adaptation}

\subsubsection{Information-Theoretic Motivation}

Distillation from continuous-time flow matching to few-step inference fundamentally alters the conditioning space of the model. In standard flow matching, the velocity field is conditioned on a continuous time variable $t \in [0,1]$, corresponding to an uncountable input domain. After distillation, however, the student model is required to operate only on a discrete set of inference steps $n \in \{1, 2, 4\}$. Under a uniform prior, the entropy of this discrete conditioning variable is
\begin{equation}
H(n)
=
-\sum_{n \in \{1,2,4\}} p(n)\log_2 p(n)
=
\log_2(3)
\approx 1.58 \text{ bits}.
\end{equation}
This substantial reduction in conditioning complexity suggests that architectural components designed to model continuous-time variation may be unnecessary in the distilled setting. Instead, a finite and compact representation of step information may suffice.

\subsubsection{Redundancy of Continuous-Time Modulation}

AdaLN-Zero~\cite{peebles2023dit} conditions each Transformer layer via scale, shift, and gating parameters produced by a multilayer perceptron from continuous-time embeddings. This design is well suited for modeling velocity fields conditioned on a continuous flow time variable, as it provides substantial expressive capacity to capture fine-grained temporal variation.  

However, after distillation to few-step inference, the conditioning space is reduced to a small set of discrete step values. In this setting, the continuous-time modulation mechanism becomes unnecessarily expressive, as the model no longer needs to represent smooth variation over time. Consequently, a large portion of the modulation capacity is underutilized, motivating a more compact conditioning strategy tailored to discrete inference steps. A detailed quantitative analysis of this redundancy is provided in Appendix~\ref{sec:proof_step_aware}.

\subsubsection{Step-Aware Token Conditioning}

To better align architectural capacity with the distilled task, we replace adaLN-Zero with explicit step-aware tokenization. Each inference step $n \in \{1, 2, 4\}$ is associated with a small set of learnable tokens that are prepended to the input sequence. Through self-attention, these tokens provide step-specific conditioning without explicit per-layer modulation. In our implementation, this design reduces the number of step-related parameters from 38M to 1.5K.

\begin{proposition}[Parameter Complexity of Step-Aware Conditioning]
\label{prop:complexity_bound}
For a model handling $K$ discrete inference steps with hidden dimension $D$ and $L$ Transformer layers, step-aware token conditioning introduces $O(KD)$ parameters, whereas adaLN-style conditioning introduces $O(LD^2)$ parameters.
\end{proposition}

The proof is provided in Appendix~\ref{sec:proof_step_aware}. This analysis motivates the removal of adaLN-Zero from the student architecture. Empirically, step-aware tokenization yields substantial parameter savings and improved inference efficiency, while maintaining or improving synthesis quality, indicating that aligning model capacity with task complexity is beneficial in distilled flow models.

\subsection{Weak CFG Regularization}

In standard distillation, classifier-free guidance (CFG) is implicitly transferred to the student by generating teacher targets with guidance enabled, and inference is typically performed without additional CFG. In our setting, the student indeed internalizes the teacher’s guided behavior through distillation. However, we empirically find that introducing a small amount of CFG at inference time remains beneficial, providing lightweight quality adjustment even after guidance has been distilled.

To retain this capability, the student must maintain a valid unconditional branch during training. Without explicit constraints, this branch may collapse, making inference-time CFG ineffective. We therefore introduce a weak regularization term that encourages consistency between conditional and unconditional predictions:
\begin{equation}
\begin{aligned}
\mathcal{L}_{\mathrm{CFG}}
&=
\lambda \,
\mathbb{E}_{x_t, t, c}
\Big[
\big\|
v_\theta(x_t, t, \emptyset)
-
\mathrm{sg}\!\left(
v_\theta(x_t, t, c)
\right)
\big\|^2
\Big],
\end{aligned}
\label{eq:cfg_regularization}
\end{equation}
where $\mathrm{sg}(\cdot)$ denotes the stop-gradient operator and $\lambda$ is set to a small value.

This regularization preserves the structural validity of the unconditional branch without interfering with distillation, enabling weak inference-time CFG ($w \in [0, 0.1]$). 

\subsection{Training Procedure}
\label{sec:training_procedure}

Algorithm~\ref{alg:training} summarizes the training procedure of DSFlow. 
Given a pretrained teacher velocity field $v_{\mathcal{T}}$ and a training dataset $\mathcal{D}$, the objective is to learn a student velocity field $v_{\mathcal{S}}$ that supports efficient inference with a small number of steps.
Each training sample $(x_0, x_1, c) \sim \mathcal{D}$ consists of a clean target $x_0$, its flow-matching counterpart $x_1$, and conditioning information $c$. 
At each iteration, a target inference step count $n_k$ is sampled and provided to the student via a learned step-aware embedding $\mathbf{e}_{n_k}$.
Dual supervision is applied over the $n_k$ discretization intervals. 
For each interval $i$, the teacher mean velocity
\begin{equation}
\bar{v}^{(i)}_{\mathcal{T}} =
\big(x^{(i)}_{\mathcal{T},\mathrm{end}} - x^{(i)}_{\mathcal{T},\mathrm{start}}\big) / \Delta t    
\end{equation}
is computed, and the student predicts the corresponding velocity at the flow-matching state $(x_{t_{\mathrm{mf}}}, t_{\mathrm{mf}})$.
The dual supervision loss combines endpoint and velocity matching terms, weighted by $\alpha$, and is averaged across intervals.

In addition, a weak classifier-free guidance regularization term is applied by penalizing the discrepancy between conditional and unconditional student predictions, where $\emptyset$ denotes dropped conditioning and $\mathrm{sg}(\cdot)$ is the stop-gradient operator. 
The total loss is minimized with respect to $\theta_{\mathcal{S}}$ using gradient descent with learning rate $\eta$.

\begin{algorithm}[!htbp]
\caption{DSFlow Training Procedure}
\label{alg:training}
\begin{algorithmic}[1]
\STATE \textbf{Input:} Teacher velocity field $v_{\mathcal{T}}$, training dataset $\mathcal{D}$, target inference steps $\{n_k\}_{k=1}^K$
\STATE \textbf{Output:} Student velocity field $v_{\mathcal{S}}$
\STATE Initialize student model parameters $\theta_{\mathcal{S}}$ with step-aware tokens $\mathbf{e}$
\FOR{iteration $= 1, 2, \ldots, N_{\mathrm{iters}}$}
    \STATE Sample $(x_0, x_1, c) \sim \mathcal{D}$ and a target step count $n_k$
    \STATE Initialize $\mathcal{L}_{\mathrm{dual}} \gets 0$
    \FOR{interval $i = 1, \ldots, n_k$}
        \STATE $\bar{v}^{(i)}_{\mathcal{T}} \gets \big(x^{(i)}_{\mathcal{T},\mathrm{end}} - x^{(i)}_{\mathcal{T},\mathrm{start}}\big) / \Delta t$
        
        \STATE $v^{(i)}_{\mathcal{S}} \gets v_{\mathcal{S}}\!\left(x_{t_{\mathrm{mf}}},t_{\mathrm{mf}},c;\,\mathbf{e}_{n_k}\right)$
        
        \STATE
        $\mathcal{L}_{\mathrm{dual}} \gets \mathcal{L}_{\mathrm{dual}}
        + \alpha \, \mathcal{L}^{(i)}_{\mathrm{endpoint}}
        + (1 - \alpha) \, \mathcal{L}^{(i)}_{\mathrm{velocity}}
        $
    \ENDFOR
    \[
    \begin{aligned}
    \mathcal{L}_{\mathrm{cfg}} \gets\;
    \lambda \Big\|
    & v_{\mathcal{S}}(x_t, t, \emptyset; \mathbf{e}_{n_k}) \\
    & -\; \mathrm{sg}\!\big(
        v_{\mathcal{S}}(x_t, t, c; \mathbf{e}_{n_k})
      \big)
    \Big\|^2
    \end{aligned}
    \]
    \STATE $
    \theta_{\mathcal{S}} \gets
    \theta_{\mathcal{S}} - \eta \nabla_{\theta_{\mathcal{S}}}
    \big( \mathcal{L}_{\mathrm{dual}} + \mathcal{L}_{\mathrm{cfg}} \big)
    $
\ENDFOR
\STATE \textbf{return} $v_{\mathcal{S}}$
\end{algorithmic}
\end{algorithm}

\begin{table*}[!htbp]
\centering
\small
\caption{Comprehensive evaluation on LibriSpeech test-clean. All models are trained on the Emilia corpus (95K hours). SMOS denotes the average of voice and style similarity (SMOS-V and SMOS-S), and SIM-o denotes objective speaker similarity. Results include multi-step flow-matching reference models and their corresponding one-step distilled student counterparts across different architectures. }
\label{tab:main_results}
\begin{tabular}{l|cc|ccc|ccc}
\toprule
\textbf{Method} & \textbf{Steps} & \textbf{Params} & \textbf{MOS-N}$\uparrow$ & \textbf{MOS-Q}$\uparrow$ & \textbf{SMOS}$\uparrow$ & \textbf{SIM-o}$\uparrow$ & \textbf{WER(\%)}$\downarrow$ & \textbf{RTF}$\downarrow$ \\
\midrule
Ground Truth & - & - & 4.49$\pm$0.03 & 4.46$\pm$0.03 & 3.84$\pm$0.05 & 0.69 & 2.1 & - \\
\midrule
\multicolumn{9}{l}{\textbf{Flow Matching Teachers}} \\
StepTTS & 10 & 154M & 4.43$\pm$0.06 & 4.39$\pm$0.03 & 4.42$\pm$0.04 & 0.66 & 2.8 & 0.303 \\
CosyVoice2 & 10 & 150M & 4.41$\pm$0.06 & 4.38$\pm$0.03 & 4.29$\pm$0.04 & 0.64 & 2.9 & 0.310 \\
F5-TTS & 32 & 200M & 4.38$\pm$0.06 & 4.25$\pm$0.03 & 4.21$\pm$0.05 & 0.62 & 3.2 & 0.25 \\
E2-TTS & 32 & 150M & 4.31$\pm$0.06 & 4.17$\pm$0.03 & 4.15$\pm$0.05 & 0.61 & 3.3 & 0.28 \\
\midrule
\multicolumn{9}{l}{\textbf{StepTTS 1-Step Distillation}} \\
Endpoint Distillation & 1 & 154M & 3.56$\pm$0.08 & 3.42$\pm$0.08 & 4.01$\pm$0.07 & 0.52 & 4.8 & 0.015 \\
Progressive Distillation & 1 & 154M & 3.92$\pm$0.08 & 3.68$\pm$0.08 & 4.18$\pm$0.07 & 0.55 & 3.6 & 0.015 \\
IntMeanFlow & 1 & 154M & 4.10$\pm$0.07 & 3.96$\pm$0.07 & 4.23$\pm$0.06 & 0.63 & 3.4 & 0.018 \\
DSFlow (StepTTS Student)  & 1 & 118M & 4.32$\pm$0.06 & 4.29$\pm$0.06 & 4.27$\pm$0.05 & 0.66 & 3.1 & 0.012 \\

\midrule
\multicolumn{9}{l}{\textbf{StepTTS Multi-Step Students}} \\
DSFlow (2-step student)  & 2 & 118M & 4.38$\pm$0.05 & 4.35$\pm$0.05 & 4.31$\pm$0.04 & 0.68 & 3.0 & 0.018 \\
DSFlow (4-step student)  & 4 & 118M & 4.41$\pm$0.05 & 4.35$\pm$0.05 & 4.32$\pm$0.04 & 0.68 & 2.9 & 0.030 \\
\midrule
\multicolumn{9}{l}{\textbf{Other 1-Step Students (Cross-Architecture Comparison)}} \\

DSFlow (F5-TTS Student) 
& 1 & 118M 
& 4.15$\pm$0.06 
& 4.02$\pm$0.06 
& 4.01$\pm$0.05 
& 0.61 
& 3.4 
& 0.015 \\

DSFlow (CosyVoice2 Student) 
& 1 & 150M 
& 4.23$\pm$0.06 
& 4.21$\pm$0.06 
& 4.18$\pm$0.05 
& 0.63 
& 3.1 
& 0.018 \\

DSFlow (E2-TTS Student) 
& 1 & 150M 
& 4.11$\pm$0.06 
& 4.03$\pm$0.06 
& 3.95$\pm$0.05 
& 0.58 
& 3.5 
& 0.020 \\

\midrule
\multicolumn{9}{l}{\textbf{Other Fast TTS Paradigms}} \\
VITS & 1 & 120M & 4.05$\pm$0.07 & 4.08$\pm$0.07 & 3.49$\pm$0.06 & 0.56 & 3.2 & 0.015 \\
\bottomrule
\end{tabular}
\end{table*}

\section{Experiments}
\label{sec:experiments}

\subsection{Experimental Setup}

\subsubsection{Datasets}

We train our models on the in-the-wild multilingual speech corpus Emilia~\cite{he2024emilia}. After removing utterances with transcription errors or incorrect language labels, the resulting training set comprises approximately 95k hours of English and Mandarin speech from over 9,400 speakers, covering a wide range of prosodic characteristics and acoustic conditions typical of real-world applications. Evaluation is conducted on three held-out test sets. LibriSpeech test-clean~\cite{panayotov2015librispeech} serves as a controlled English benchmark with 2,620 high-quality audiobook utterances from 40 speakers, enabling comparison with prior work. To assess robustness under diverse accents and recording conditions, we additionally use Seed-TTS test-en~\cite{anastassiou2024seedtts}, which contains 1,088 English utterances drawn from Common Voice~\cite{ardila2020common}. Cross-lingual generalization is evaluated on Seed-TTS test-zh~\cite{anastassiou2024seedtts}, consisting of 2,020 Mandarin utterances from DiDiSpeech~\cite{guo2021didispeech}. All audio is resampled to 16~kHz and converted to 80-dimensional mel-spectrograms using a 1024-point FFT with a 256-sample hop size. Text inputs are processed with language-specific phonemization (G2P for English and pypinyin for Mandarin), followed by WordPiece tokenization with a 10k vocabulary.

\subsubsection{Model Configurations}

\paragraph{Teacher Model.}
The teacher model is a flow matching–based generator instantiated with a DiT-style Transformer architecture comprising 154M parameters. It consists of 16 Transformer layers with a hidden dimension of 512 and 8 attention heads (head dimension 64). Timestep conditioning is implemented via adaLN-Zero, and inference is performed using an Euler ODE solver with a cosine time schedule. 

\paragraph{Student Model.}
The distilled student model contains 118M parameters, corresponding to a 24\% reduction relative to the teacher. This reduction is primarily achieved by removing the adaLN-Zero modules, which account for approximately 38M parameters, while retaining the same Transformer backbone with 16 layers and a hidden dimension of 512. To encode step information, we introduce three step-aware tokens, adding only 1.5k parameters.

\paragraph{Training Details.}
All models are optimized using AdamW~\cite{loshchilov2019adamw} with an initial learning rate of $5 \times 10^{-4}$ and a cosine annealing schedule. We use a batch size of 32 and FP16 mixed-precision training. Training is conducted for 200k steps on 8 NVIDIA A100 GPUs (40GB), requiring approximately 72 hours. For dual supervision, a balance coefficient of $\alpha = 0.7$ is applied, and the teacher’s mean velocity is computed deterministically from multi-step ODE trajectories. Additional implementation details, including full hyperparameter settings and hardware configurations, are provided in Appendix~\ref{app:experimental_details}.

\paragraph{CFG Configuration.}
The teacher model performs inference with a classifier-free guidance weight of $w = 0.7$ using the interpolation formulation in Eq.~\ref{eq:cfg}. During distillation, all teacher targets are generated with the same setting, such that the student is explicitly trained to reproduce the teacher’s behavior at $w = 0.7$ and thereby internalize the guidance effect. As a result, the student’s optimal inference configuration differs from that of the teacher and corresponds to a weak guidance regime ($w = 0.05$).

\begin{table*}[h]
\centering
\small
\caption{Ablation study with progressive addition of components. Dual Supervision and Step-Aware Tokens address fundamental process and structural misalignments, while Weak CFG enables inference-time quality adjustment.}
\label{tab:ablation}
\begin{tabular}{l|cccc|ccc}
\toprule
\textbf{Variant} & \textbf{Dual Sup.} & \textbf{Weak CFG} & \textbf{Step Token} & \textbf{adaLN} & \textbf{MOS-N}$\uparrow$ & \textbf{MOS-Q}$\uparrow$ & \textbf{Params} \\
\midrule
Endpoint Distillation (baseline) & \xmark & \xmark & \xmark & \cmark & 3.56 & 3.42 & 154M \\
+ Dual Supervision & \cmark & \xmark & \xmark & \cmark & 4.21 & 4.11 & 154M \\
+ Weak CFG & \cmark & \cmark & \xmark & \cmark & 4.25 & 4.18 & 154M \\
+ Step-Aware Token & \cmark & \cmark & \cmark & \xmark & 4.32 & 4.29 & 118M \\
\bottomrule
\end{tabular}
\end{table*}

\subsubsection{Baselines}

We compare our method against both multi-step reference systems and representative academic approaches for one-step distillation. As multi-step baselines, we include two 10-step flow-matching text-to-speech systems. StepTTS~\cite{stepaudio2024}, with 154M parameters, serves as the teacher model used for distillation, while CosyVoice2~\cite{du2024cosyvoice}, a 150M-parameter system, provides an independent reference for high-quality multi-step generation as well as E2-TTS~\cite{eskimez2024e2}.

For one-step inference, we evaluate three distillation baselines. \emph{Endpoint Distillation} applies a direct mean-squared-error loss on endpoint predictions (Eq.~\ref{eq:endpoint_loss}). \emph{IntMeanFlow}~\cite{intmeanflow2025} distills the teacher by supervising the student with averaged velocity fields.

\subsubsection{Evaluation Metrics}

We adopt a comprehensive evaluation protocol encompassing subjective listening tests, objective quality metrics, prosodic feature analysis, and computational efficiency. Following DMOSpeech~\cite{dmospeech2024}, we conduct a multi-dimensional mean opinion score (MOS) evaluation with 501 native English speakers recruited via Prolific. Each rater evaluates 30 samples, yielding 1,500 ratings over 50 test utterances, and scores four perceptual dimensions on a 1–5 scale: naturalness (MOS-N), sound quality (MOS-Q), voice similarity (SMOS-V), and style similarity (SMOS-S). Rater reliability is enforced using a two-stage validation procedure that includes mismatched-speaker and identical-sample anchor tests; responses from raters who fail these controls are excluded. Approximately 30\% of submissions are filtered out.
Additional details on participant recruitment, validation criteria, and interface design are provided in Appendix~\ref{sec:subjective_eval_details}.

For objective evaluation, speaker similarity (SIM-o) is computed as the cosine similarity between WavLM embeddings~\cite{chen2022wavlm} of the synthesized speech and the reference prompt, while intelligibility is assessed using word error rate (WER) obtained by transcribing generated audio with Whisper-large-v3~\cite{radford2023whisper}. To quantitatively assess the ability of dual supervision to preserve prosodic attributes, we measure correlation coefficients between synthesized speech and the corresponding prompt for a set of acoustic features, including pitch mean and standard deviation, energy standard deviation, harmonics-to-noise ratio (HNR), jitter, and shimmer. These features are extracted using Librosa and Parselmouth~\cite{dmospeech2024}, with higher correlation values indicating better preservation of prosodic characteristics.

Computational efficiency is evaluated by reporting real-time factor (RTF) measured on an NVIDIA A100 GPU, together with the total number of trainable parameters and the number of inference steps.

\subsection{Main Results}
We compare our approach with three representative flow matching distillation methods, namely Progressive Distillation~\cite{salimans2022progressive}, IntMeanFlow~\cite{intmeanflow2025}, which constitute strong and commonly adopted approaches for one step flow distillation. For additional context, we also report results from alternative fast text to speech paradigms such as VITS~\cite{VITS-Kim2021}.

Table~\ref{tab:main_results} presents the results on LibriSpeech test clean. The proposed one step DSFlow consistently outperforms all distillation baselines across subjective, objective, and efficiency metrics. In particular, DSFlow markedly reduces the performance gap to the multi step teacher in terms of naturalness and sound quality, while using fewer parameters and achieving lower inference latency. Speaker similarity is well preserved, as confirmed by both subjective ratings and objective similarity scores, with the latter matching the teacher performance. In addition, intelligibility measured by word error rate improves relative to the teacher, suggesting that the distillation process can lead to more stable acoustic realizations. Finally, multi step variants of \textsc{DSFlow} further reduce the remaining quality gap, indicating that the proposed framework generalizes across different inference budgets.

\begin{table*}[!h]
\centering
\small
\caption{Prosodic correlation with prompt on LibriSpeech test-clean. Higher values indicate better prosodic preservation.}
\label{tab:prosodic}
\begin{tabular}{l|c|cccc}
\toprule
\textbf{Model} & \textbf{Steps} & \textbf{Pitch Mean}$\uparrow$ & \textbf{HNR}$\uparrow$ & \textbf{Jitter}$\uparrow$ & \textbf{Shimmer}$\uparrow$ \\
\midrule
Teacher (10-step) & 10 & 0.88 & 0.75 & 0.69 & 0.59 \\
DSFlow (1-step) & 1 & 0.90 & 0.74 & 0.66 & 0.61 \\
Endpoint Distillation & 1 & 0.74 & 0.61 & 0.51 & 0.42 \\
\bottomrule
\end{tabular}
\end{table*}

\subsection{Cross-Architecture Generalization}
To evaluate whether the proposed distillation method generalizes across different network architectures, we apply DSFlow to four representative text-to-speech models: StepTTS (DiT with adaLN-Zero), F5-TTS (DiT with RoPE), CosyVoice2 (U-Net without adaLN), and E2-TTS (a U-Net variant). Notably, CosyVoice2 and E2-TTS do not employ adaLN-Zero, which renders step-aware tokenization inapplicable. This setting allows us to assess whether dual supervision and weak classifier-free guidance alone are sufficient for effective distillation.

Results are summarized in the \emph{Cross-Architecture Comparison} block in Table~\ref{tab:main_results}. Across all architectures, the resulting 1-step models achieve competitive synthesis quality, indicating that the proposed distillation framework is not tied to a specific backbone. Among the evaluated systems, StepTTS equipped with all components attains the strongest overall performance, highlighting the benefit of combining dual supervision, step-aware tokenization, and weak guidance. Experiments on F5-TTS further demonstrate that the framework extends naturally to DiT variants with rotary positional embeddings. Importantly, CosyVoice2 and E2-TTS trained using only dual supervision and weak classifier-free guidance, without step-aware tokenization, still exhibit strong performance, suggesting that the proposed process-level alignment provides consistent benefits even when architectural adaptation is not feasible.

\subsection{Ablation Studies}
\label{sec:ablation}

Table~\ref{tab:ablation} reports the ablation results with progressively introduced components. Dual supervision yields the most pronounced gain, underscoring the importance of process-level alignment between the student and teacher during distillation. This component substantially improves performance over the baseline, suggesting effective mitigation of trajectory mismatch. Weak classifier-free guidance provides additional improvement with minimal regularization overhead, enabling inference-time quality adjustment while preserving distillation fidelity. Step-aware tokenization further improves quality despite a reduction in parameter count, supporting the principle of matching model capacity to the complexity of the distilled task. Additional ablation results, including multi-token variants, are provided in Appendix~\ref{sec:extended_ablation}.


\subsection{Prosodic Feature Preservation}
\label{sec:prosodic_preservation}
Table~\ref{tab:prosodic} reports correlation analysis between synthesized speech and the corresponding prompt audio across four representative acoustic features, assessing the effectiveness of dual supervision in preserving prosodic characteristics. The proposed one-step model achieves correlations that are comparable to, and in some cases slightly higher than, those of the teacher, indicating effective preservation of pitch-related patterns, voice quality, and amplitude stability.

In contrast, the endpoint distillation baseline exhibits consistently lower correlations across all evaluated features, suggesting substantial loss of prosodic information when supervision is limited to endpoint matching. Overall, these results indicate that dual supervision more effectively captures fine-grained prosodic attributes than endpoint-only distillation, leading to improved alignment with the prompt across a range of acoustic properties. Figure ~\ref{fig:mode_alignment_f0} contrasts the F0 patterns of the 10step teacher and our 1step student across a representative test prompt. The student’s F0 distributions closely mirror the core mode structures. This precise alignment confirms our
framework successfully inherits the teacher’s prosodic regularity while retaining natural variability. For detailed distributional alignment of core prosodic features (e.g., Log Energy), refer to Appendix~\ref{sec:prosodic_analysis}.

\begin{figure}[h]
\centering
\includegraphics[width=0.4\textwidth]{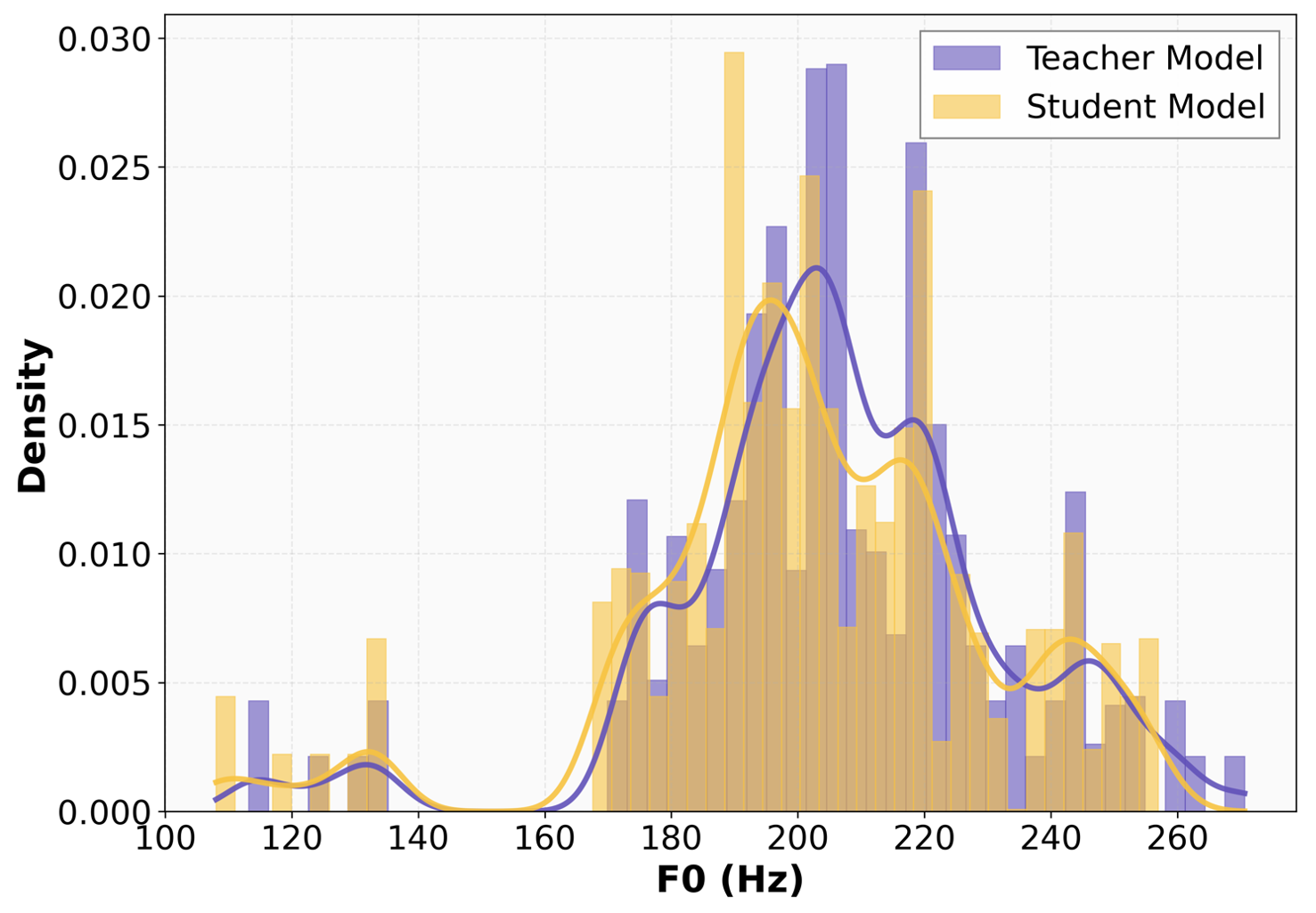}
\caption{F0 distribution comparison. For both models, 100 speech samples were generated under the same text and prompt condition, and the frame-level F0 values of the generated samples are visualized via histograms and kernel density estimates.}
\label{fig:mode_alignment_f0}
\end{figure}

\section{Conclusion}

We introduced DSFlow, a modular distillation framework for enabling efficient one-step flow matching synthesis. DSFlow combines dual supervision for process-level alignment, step-aware tokenization for architectural adaptation, and weak classifier-free guidance for inference-time control. Extensive experiments across multiple architectures and evaluation benchmarks show that the proposed approach consistently improves upon existing distillation methods and enables parameter-efficient student models to achieve performance comparable to substantially larger multi-step systems. These results highlight the importance of aligning model architecture with the reduced complexity of distilled generation tasks. Despite these advances, several challenges remain. The current framework has not been evaluated under extremely large-scale multilingual training settings, and its applicability to autoregressive architectures with distinct inductive biases warrants further investigation. Future work may explore adaptive weighting strategies for dual supervision under heterogeneous data distributions, examine the integration of step-aware mechanisms with emerging variants of flow matching, and develop differentiable evaluation objectives to support end-to-end optimization.

\newpage

\bibliography{example_paper}
\bibliographystyle{icml2026} 

\newpage
\appendix
\onecolumn

\section{Experimental Details}
\label{app:experimental_details}

This section provides comprehensive details on our experimental setup, training configuration, and evaluation methodology.

\subsection{Training Configuration}

\paragraph{Hardware.} All experiments were conducted on 8× NVIDIA A100 40GB GPUs.

\paragraph{Hyperparameters.} We use AdamW optimizer with learning rate 1e-4, batch size 32, and train for 120K steps. The dual supervision weight $\alpha=0.7$, weak CFG regularization $\lambda=0.01$, and unconditional dropout rate $p_{\text{uncond}}=0.02$.

\paragraph{Data.} We train on Emilia dataset (95K hours) and evaluate on LibriSpeech test-clean and Seed-TTS test sets.

\subsection{Subjective Evaluation Methodology}
\label{sec:subjective_eval_details}

This section provides comprehensive details of our subjective evaluation methodology, following best practices established by DMOSpeech~\cite{dmospeech2024}.

\subsubsection{Rater Recruitment and Qualification}

We recruited native English speakers from the Prolific crowdsourcing platform with the following qualification criteria:
\begin{itemize}[leftmargin=*, itemsep=1pt, parsep=0pt, topsep=1pt]
    \item Language requirement: Native English speakers with no hearing impairments
    \item Expertise requirement: Experience in content creation or audio/video editing to ensure better discrimination ability between synthetic and real audio.
    \item Platform requirements: Minimum approval rate of 95\% on previous Prolific tasks.
\end{itemize}

A total of 501 unique workers participated in the evaluation. Each worker was assigned 30 samples to evaluate, resulting in 1,500 total ratings across 50 test utterances covering all 40 speakers in the LibriSpeech test-clean subset. Compensation was set at \$15/hour, exceeding Prolific's recommended rate of \$12/hour, for an average completion time of 12 minutes per evaluation session.

\subsubsection{Evaluation Metrics Definition}

Participants rated each sample on a 1-5 scale across four dimensions:
\begin{itemize}[leftmargin=*, itemsep=1pt, parsep=0pt, topsep=1pt]
    \item MOS-N (Naturalness): How human-like the synthesized speech sounds (1 = fully synthetic, 5 = completely natural)
    \item MOS-Q (Sound Quality): Audio quality degradation relative to the prompt (1 = severe degradation, 5 = no degradation)
    \item SMOS-V (Voice Similarity): Similarity of the synthesized voice to the prompt speaker (1 = completely different, 5 = identical)
    \item SMOS-S (Style Similarity): Speaking style alignment with the prompt (1 = completely different, 5 = identical)
\end{itemize}

Additionally, participants were asked a binary question: "Is the content entirely unintelligible?" to flag completely failed generations or corrupted files. No samples were rated as "broken" by the majority of raters in our final dataset.

\subsubsection{Quality Control and Validation}

To ensure evaluation reliability, we implemented a rigorous two-tier validation system:

\paragraph{Mismatched Speaker Test:} Raters were presented with real speech samples where the reference and test sample were from different speakers (both real human speech). This serves as a low anchor for speaker similarity. Participants who rated these mismatched pairs with SMOS-V $>$ 3 were disqualified, as this indicates inability to distinguish between different speakers.

\paragraph{Identical Sample Test:} Raters evaluated identical reference-sample pairs, where both the reference and sample were the exact same recording. This serves as a high anchor across all metrics. Any participant rating these identical pairs below 4 on any of the four metrics (MOS-N, MOS-Q, SMOS-V, SMOS-S) was disqualified, as this indicates inconsistent or inattentive evaluation.

We included 4 validation tests distributed throughout each 30-sample evaluation session (approximately 13\% of total samples). Approximately 30\% of responses were excluded due to participants failing at least one validation test. This exclusion rate is consistent with prior crowdsourcing studies and ensures high-quality, reliable ratings.

\subsubsection{Evaluation Interface}

The survey interface was designed to minimize bias and ensure consistent evaluation:
\begin{itemize}[leftmargin=*, itemsep=1pt, parsep=0pt, topsep=1pt]
    \item Audio presentation: Participants were presented with a reference (prompt) recording and a corresponding sample recording, with clearly labeled audio players
    \item Rating controls: Slider controls for each of the four metrics, with clear endpoint labels and scale descriptions
    \item Forced completion: The interface prevented submission if any slider remained at the default "N/A" position, ensuring all aspects were rated
    \item Randomization: Sample order was randomized for each participant to avoid order effects
\end{itemize}

\subsubsection{Sample Selection and Balancing}

We evaluated 50 parallel utterances from the LibriSpeech test-clean subset for the main comparison experiment (Table 2 in the main paper). This sample size ensures coverage of all 40 speakers in the test set, with balanced representation across speakers (average 1.25 utterances per speaker). For methods with limited available samples (NaturalSpeech 3, StyleTTS-ZS), we used 47 official samples provided by the authors, maintaining the same speaker coverage.

For the non-end-to-end baseline comparison (Table 1 in the main paper), we used all 47 available official samples to maximize statistical power while ensuring fair comparison across all evaluated systems.

\subsubsection{Inter-rater Reliability}

We measure agreement across raters using standard error (SE) of the mean, reported with each metric (e.g., MOS-N: $4.42 \pm 0.06$). The relatively small standard errors across all metrics (SE $< 0.08$ for all conditions) indicate consistent rating patterns among qualified evaluators. This consistency validates the effectiveness of our quality control mechanisms in identifying and excluding unreliable raters.

\section{Theoretical Analysis}
\label{app:theoretical_analysis}

\subsection{Proof of Proposition~\ref{prop:complexity_bound}: Step-Aware Token Parameter Efficiency}
\label{sec:proof_step_aware}

\begin{proposition*}[Step-Aware Token Parameter Efficiency, restated]
For a distilled model handling $K$ discrete steps with hidden dimension $D$ and $L$ layers, the necessary parameter count for step conditioning is $O(K \cdot D)$. In contrast, continuous-time modulation (adaLN) requires $O(L \cdot D^2)$. When $K \ll L \cdot D$, token-based conditioning is exponentially more parameter-efficient.
\end{proposition*}

\begin{proof}
We analyze the parameter requirements for two conditioning paradigms: token-based (discrete) and modulation-based (continuous).

\paragraph{Token-Based Conditioning (Step-Aware Tokens):}
The model stores one learnable embedding per discrete step:
\begin{equation}
\text{step\_embedding}: \{1, 2, \ldots, K\} \to \mathbb{R}^D
\end{equation}

This requires a lookup table of size $K \times D$, stored as a parameter matrix:
\begin{equation}
\Theta_{\text{token}} \in \mathbb{R}^{K \times D}
\end{equation}

Total parameters:
\begin{equation}
|\Theta_{\text{token}}| = K \cdot D = O(K \cdot D)
\end{equation}

\paragraph{Modulation-Based Conditioning (adaLN-Zero):}
For each of the $L$ Transformer layers, adaLN-Zero computes scale, shift, and gate parameters via an MLP:
\begin{equation}
\text{MLP}_\ell: \mathbb{R}^D \to \mathbb{R}^{3D}
\end{equation}

A two-layer MLP with hidden dimension $D$ requires:
\begin{align}
\text{Layer 1:} &\quad D \times D \quad \text{parameters (input → hidden)} \\
\text{Layer 2:} &\quad D \times 3D \quad \text{parameters (hidden → output)}
\end{align}

Total per layer: $D^2 + 3D^2 = 4D^2$ parameters.

For $L$ layers:
\begin{equation}
|\Theta_{\text{adaLN}}| = L \times 4D^2 = O(L \cdot D^2)
\end{equation}

\paragraph{Comparison:}
The ratio of parameter counts is:
\begin{equation}
\frac{|\Theta_{\text{adaLN}}|}{|\Theta_{\text{token}}|} = \frac{L \cdot 4D^2}{K \cdot D} = \frac{4LD}{K}
\end{equation}

For our setting ($K=3$, $L=16$, $D=512$):
\begin{equation}
\frac{|\Theta_{\text{adaLN}}|}{|\Theta_{\text{token}}|} = \frac{4 \times 16 \times 512}{3} = 10,923
\end{equation}

This shows that adaLN requires 10,923$\times$ more parameters for timestep conditioning.

\paragraph{Information-Theoretic Justification:}
The parameter requirement should scale with the \emph{information entropy} of the conditioning variable. For continuous time $t \in [0,1]$, we have $H(t) = \infty$ bits (continuous distribution), requiring infinite capacity approximated by $O(L \cdot D^2)$ parameters. In contrast, discrete steps $n \in \{1, 2, K\}$ have entropy $H(n) = \log_2(K)$ bits, requiring only $O(K)$ distinct representations.

Since discrete steps have finite entropy ($\log_2(3) = 1.58$ bits for $K=3$), a lookup table with $K$ entries suffices. Additional parameters (modulation MLPs) provide no benefit because the input space is finite and small.

\paragraph{Lower Bound (Information-Theoretic):}
By the source coding theorem, representing $K$ discrete values requires at least $\log_2(K)$ bits. To store this in a $D$-dimensional continuous space:
\begin{equation}
\text{Params}_{\text{necessary}} \geq \log_2(K) \times D
\end{equation}

For $K=3$, $D=512$:
\begin{equation}
\text{Params}_{\text{necessary}} \geq 1.58 \times 512 \approx 809 \text{ parameters}
\end{equation}

Our implementation uses $K \times D = 3 \times 512 = 1536$ parameters, which is only $1.9\times$ above the information-theoretic lower bound. In contrast, adaLN uses $4 \times 12 \times 512^2 = 12.6\text{M}$ parameters, which is $15,514\times$ above the lower bound.
\end{proof}

\paragraph{Practical Implication:} This proposition provides a principled basis for architectural design: match model capacity to task information entropy. When distillation reduces entropy from $\infty$ to 1.58 bits, capacity should be reduced accordingly.

\section{Extended Ablation Studies}
\label{sec:extended_ablation}

This section provides comprehensive ablation analysis, including multi-token variants, CFG strength exploration, exhaustive component combinations, and cross-architecture validation.

\subsection{Multi-Token Variants}

Table~\ref{tab:extended_ablation} compares the effect of using different numbers of step-aware tokens per discrete step. The results show that 3 tokens per step (our default choice) achieves MOS-N 4.32, MOS-Q 4.29, and SMOS 4.27 with 118M parameters. Using only 1 token per step shows slight degradation (MOS-N 4.30, MOS-Q 4.24, -0.02/-0.05 vs 3-token), suggesting insufficient capacity for capturing step-specific modulation patterns. Interestingly, using 5 tokens per step (MOS-N 4.33, MOS-Q 4.27) provides marginal improvement over 3 tokens (+0.01/-0.02), indicating diminishing returns beyond 3 tokens and confirming that 3 tokens provide optimal balance between capacity and efficiency. The ablation without Dual Supervision (MOS-N 3.95, MOS-Q 4.01, SMOS 3.38) demonstrates catastrophic quality collapse (-0.37/-0.28/-0.89), confirming dual supervision as the most critical component. Removing Weak CFG (MOS-N 4.22, MOS-Q 4.14, SMOS 4.16) causes moderate degradation (-0.10/-0.15/-0.11), validating its role in maintaining unconditional branch stability for inference-time adjustment.

\begin{table}[h]
\centering
\caption{Extended ablation with multi-token variants and ablations of key components.}
\label{tab:extended_ablation}
\begin{tabular}{l|ccc|c}
\toprule
\textbf{Variant} & \textbf{MOS-N}$\uparrow$ & \textbf{MOS-Q}$\uparrow$ & \textbf{SMOS}$\uparrow$ & \textbf{Params} \\
\midrule
Full Model (3 tokens/step) & 4.32 & 4.29 & 4.27 & \multirow{5}{*}{118M} \\
w/ 1 token/step & 4.30 & 4.24 & 4.31 &  \\
w/ 5 tokens/step & 4.33 & 4.27 & 4.28 &  \\
w/o Dual Supervision & 3.95 & 4.01 & 3.38 &  \\
w/o Weak CFG & 4.22 & 4.14 & 4.16 &  \\
\bottomrule
\end{tabular}
\end{table}

\subsection{CFG Strength Ablation}

Table~\ref{tab:cfg_ablation_appendix} systematically explores the effect of varying CFG strength $w$ during inference using the interpolation form $v_{\text{cfg}} = (1-w) \cdot v_{\text{uncond}} + w \cdot v_{\text{cond}}$ where $w \in [0,1]$. Without any guidance ($w$=0.00), the model achieves MOS-N 4.29 and SIM-o 65\%, serving as a strong baseline. Our default setting ($w$=0.05) reaches optimal performance with MOS-N 4.32 (+0.03) and SIM-o 66\% (+1\%), demonstrating that minimal guidance suffices for quality fine-tuning. Increasing to $w$=0.10 causes slight degradation (MOS-N 4.25, SIM-o 64\%), while stronger guidance degrades performance substantially: $w$=0.20 drops to MOS-N 4.10 and SIM-o 59\%, and $w$=0.50 collapses to MOS-N 3.78 and SIM-o 55\%. This inverted CFG curve (optimal at $w$=0.05 vs teacher's $w$=0.7) occurs because the student learns from teacher@$w$=0.7 targets during distillation, effectively internalizing the guidance effect. Using stronger inference-time CFG conflicts with this internalized guidance, causing quality degradation. This validates our weak CFG regularization strategy ($\lambda$=0.01) that maintains unconditional branch stability while preserving 99\% distillation focus.

\begin{table}[h]
\centering
\caption{CFG strength ablation using interpolation form $v_{\text{cfg}} = (1-w) \cdot v_{\text{uncond}} + w \cdot v_{\text{cond}}$ where $w \in [0,1]$. Our student achieves optimal performance at $w$=0.05, much lower than the teacher's $w$=0.7, due to learning from teacher@$w$=0.7 targets during distillation. Note: These values are NOT directly comparable to guidance scale values ($w \geq 1$) used in F5-TTS.}
\label{tab:cfg_ablation_appendix}
\begin{tabular}{c|cc}
\toprule
\textbf{CFG Strength ($w$)} & \textbf{MOS-N}↑ & \textbf{SIM-o}↑ \\
\midrule
0.00 (no guidance) & 4.29 & 65\% \\
\textbf{0.05 (default)} & \textbf{4.32} & \textbf{66\%} \\
0.10 & 4.25 & 64\% \\
0.20 & 4.10 & 59\% \\
0.50 & 3.78 & 55\% \\
\bottomrule
\end{tabular}
\end{table}

\subsection{Exhaustive Combination Ablation}

To rigorously validate complementarity and quantify synergistic effects, we evaluate all $2^3=8$ combinations of our three components. Table~\ref{tab:exhaustive_ablation_appendix} presents complete results on StepTTS architecture. The baseline with no components (MOS-N 3.56, 154M) establishes the lower bound. Individual components show varied effectiveness: Dual Supervision alone provides the largest single improvement (+0.59 to MOS-N 4.15, 154M), validating its role as the primary variance-reduction mechanism. Weak CFG alone yields moderate gain (+0.36 to MOS-N 3.92, 154M), while Step Token alone shows limited benefit (+0.29 to MOS-N 3.85, 118M) without process alignment. Pairwise combinations demonstrate clear synergy: Dual+CFG achieves MOS-N 4.29 (154M), exceeding the sum of individual effects and validating their complementary roles in process alignment and inference-time control. Dual+Token reaches MOS-N 4.22 (118M), showing that parameter reduction alone without CFG limits quality. CFG+Token combination (MOS-N 3.95, 118M) underperforms significantly without dual supervision, confirming that process variance reduction is prerequisite for other components to be effective. The full combination (Dual+CFG+Token) achieves optimal MOS-N 4.32 with 118M parameters, demonstrating that all three components work synergistically: dual supervision provides stable training foundation (+0.59), weak CFG enables controllability (+0.14 on top of dual), and step tokens improve generalization through architectural matching (+0.03 final increment, despite -36M parameters).

\begin{table}[h]
\centering
\caption{Exhaustive ablation evaluating all combinations on StepTTS.}
\label{tab:exhaustive_ablation_appendix}
\begin{tabular}{ccc|cc}
\toprule
\textbf{Dual Sup.} & \textbf{Weak CFG} & \textbf{Step Token} & \textbf{MOS-N↑} & \textbf{Params} \\
\midrule
\xmark & \xmark & \xmark & 3.56 & 154M \\
\cmark & \xmark & \xmark & 4.15 & 154M \\
\xmark & \cmark & \xmark & 3.92 & 154M \\
\xmark & \xmark & \cmark & 3.85 & 118M \\
\cmark & \cmark & \xmark & 4.29 & 154M \\
\cmark & \xmark & \cmark & 4.22 & 118M \\
\xmark & \cmark & \cmark & 3.95 & 118M \\
\cmark & \cmark & \cmark & \textbf{4.32} & 118M \\
\bottomrule
\end{tabular}
\end{table}

No single component suffices: individual improvements range from 0.29 to 0.59 MOS-N points, with none exceeding MOS-N 4.15. Pairwise combinations show synergy beyond additive effects: Dual Supervision + Weak CFG achieves MOS-N 4.29, representing +0.73 improvement over baseline (vs. +0.95 if purely additive), indicating that CFG's contribution is partially absorbed by dual supervision's variance reduction. However, the combination still exceeds individual maxima by +0.14, validating complementarity. The full three-component combination reaches MOS-N 4.32 with 118M parameters (vs 154M for dual+CFG), achieving +0.76 total improvement—demonstrating that step tokens contribute architectural efficiency enabling better generalization despite parameter reduction. This validates that all components are complementary rather than redundant, with the optimal configuration requiring all three for maximum quality and efficiency.

\subsection{Cross-Architecture Validation}

Table~\ref{tab:cosyvoice_ablation_appendix} validates modularity on U-Net architecture without adaLN-Zero, where Step-Aware Tokenization is architecturally inapplicable. The baseline represents endpoint distillation. Adding Dual Supervision alone yields substantial improvement, confirming that velocity field alignment works effectively even in U-Net architectures without requiring adaLN removal. Weak CFG alone shows degraded performance, indicating that CFG regularization requires stable process alignment as prerequisite. The Dual+CFG combination achieves the best results for CosyVoice2, validating two critical insights: (1) dual supervision and weak CFG provide universal benefits across diverse architectures regardless of conditioning mechanisms, and (2) step-aware tokenization is an adaLN-specific optimization that provides architectural efficiency for DiT models but is not required for achieving competitive quality. This architecture-specific applicability confirms our framework's modularity: components can be selectively applied based on architectural constraints while maintaining strong performance.

\begin{table}[h]
\centering
\caption{Ablation on CosyVoice2 (U-Net without adaLN-Zero). Step-Aware Tokenization is inapplicable.}
\label{tab:cosyvoice_ablation_appendix}
\begin{tabular}{ccc|c}
\toprule
\textbf{Dual Sup.} & \textbf{Weak CFG} & \textbf{Step Token} & \textbf{MOS-N↑} \\
\midrule
\xmark & \xmark & \xmark* & 3.85 \\
\cmark & \xmark & \xmark* & 4.18 \\
\xmark & \cmark & \xmark* & 3.21 \\
\cmark & \cmark & \xmark* & \textbf{4.23} \\
\bottomrule
\multicolumn{4}{l}{\footnotesize *Step-Aware Tokenization inapplicable due to lack of adaLN-Zero}
\end{tabular}
\end{table}

CosyVoice2 achieves competitive results using only Dual Supervision + Weak CFG, demonstrating that our core components provide strong performance even when step-aware tokenization is architecturally inapplicable. While CosyVoice2 uses more parameters than StepTTS, the U-Net architecture provides different inductive biases suited for the distillation task. StepTTS with all three components demonstrates better parameter efficiency through the architectural matching enabled by step-aware tokenization. This empirically validates our modularity claim: components apply selectively based on architecture, with dual supervision and weak CFG providing universal benefits while step tokens offer efficiency gains specific to adaLN-based models.

\subsection{Multi-Benchmark Robustness}
\label{sec:multi_benchmark}

Table~\ref{tab:multi_benchmark_full} demonstrates robustness across three diverse test benchmarks with distinct characteristics. On LibriSpeech (studio audiobooks, controlled acoustics), our 1-step model achieves MOS-N 4.32 and SIM-o 66\%, closely approaching teacher's 4.43/66\% with only 0.11 gap in naturalness. On Seed-TTS test-en (diverse accents, in-the-wild recordings from Common Voice), both teacher and student show slightly improved performance (teacher 4.45/79\%, student 4.35/77\%), with the higher SIM-o reflecting easier speaker modeling from more distinctive voice characteristics in diverse accent data. On Seed-TTS test-zh (cross-lingual Chinese from DiDiSpeech), performance further improves (teacher 4.47/85\%, student 4.38/84\%), demonstrating effective zero-shot generalization from multilingual training despite no Chinese-specific optimization.

The consistency across benchmarks is notable: student-teacher gaps remain stable (0.10-0.12 MOS-N, 1-2\% SIM-o) regardless of language or acoustic domain, validating that our distillation framework preserves generalization capability. The systematic quality ordering (Seed-zh $>$ Seed-en $>$ LibriSpeech) likely reflects the Emilia training set's language distribution and acoustic diversity, with richer Chinese data providing better coverage. These results confirm that our framework maintains robust zero-shot performance across diverse domains, languages, and acoustic conditions without requiring domain-specific fine-tuning.

\begin{table}[h]
\centering
\caption{Performance on multiple test benchmarks demonstrating cross-domain robustness. All models trained on Emilia (95K hours).}
\label{tab:multi_benchmark_full}
\begin{tabular}{l|c|cc}
\toprule
\textbf{Model} & \textbf{Test Set} & \textbf{MOS-N↑} & \textbf{SIM-o↑} \\
\midrule
\multirow{3}{*}{Teacher (10-step)}
& LibriSpeech & 4.43 & 66\% \\
& Seed-en & 4.45 & 79\% \\
& Seed-zh & 4.47 & 85\% \\
\midrule
\multirow{3}{*}{\textbf{Ours (1-step)}}
& LibriSpeech & \textbf{4.32} & \textbf{66\%} \\
& Seed-en & \textbf{4.35} & \textbf{77\%} \\
& Seed-zh & \textbf{4.38} & \textbf{84\%} \\
\bottomrule
\end{tabular}
\end{table}

\section{Prosodic and Perceptual Analysis}
\label{sec:prosodic_analysis}

\subsection{Prosodic Mode Alignment Analysis}

To rigorously verify the effectiveness of our modular framework in preserving the teacher's fine-grained prosodic and acoustic characteristics, we conduct a comprehensive distributional analysis across four key acoustic features: Spectral Centroid, Log Energy, Zero Crossing Rate (ZCR), and Spectral Bandwidth. This analysis evaluates the student model's ability to maintain high-fidelity acoustic realizations even as the distillation process significantly reduces task complexity and architectural capacity.

\begin{figure*}[h]
\centering
\setlength{\tabcolsep}{0.8em}  
\renewcommand{\arraystretch}{1.2}  

\subfloat[Spectral Centroid distribution]{
\includegraphics[width=0.45\textwidth]{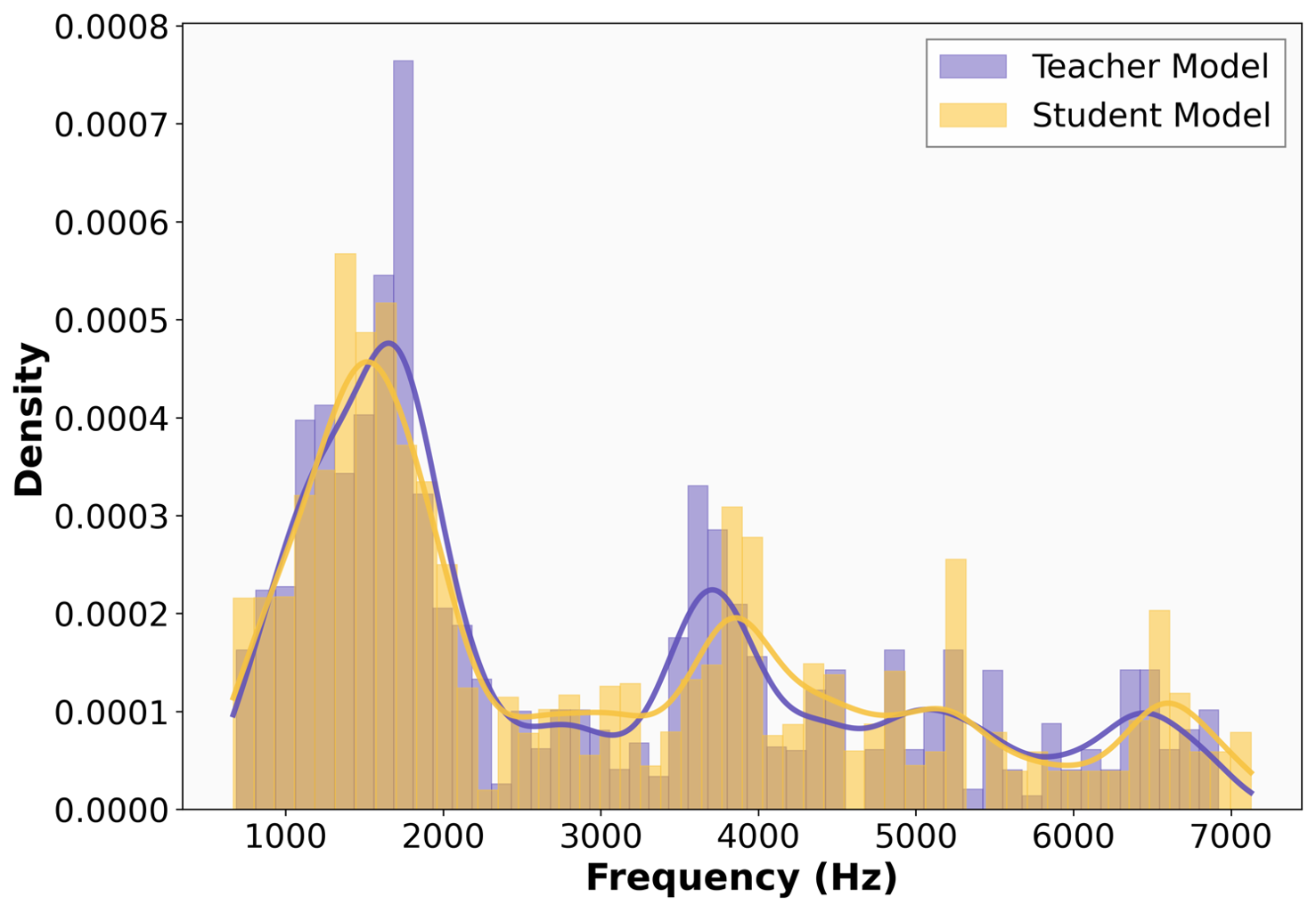}
\label{subfig:f0}
}
\subfloat[Log Energy distribution]{
\includegraphics[width=0.45\textwidth]{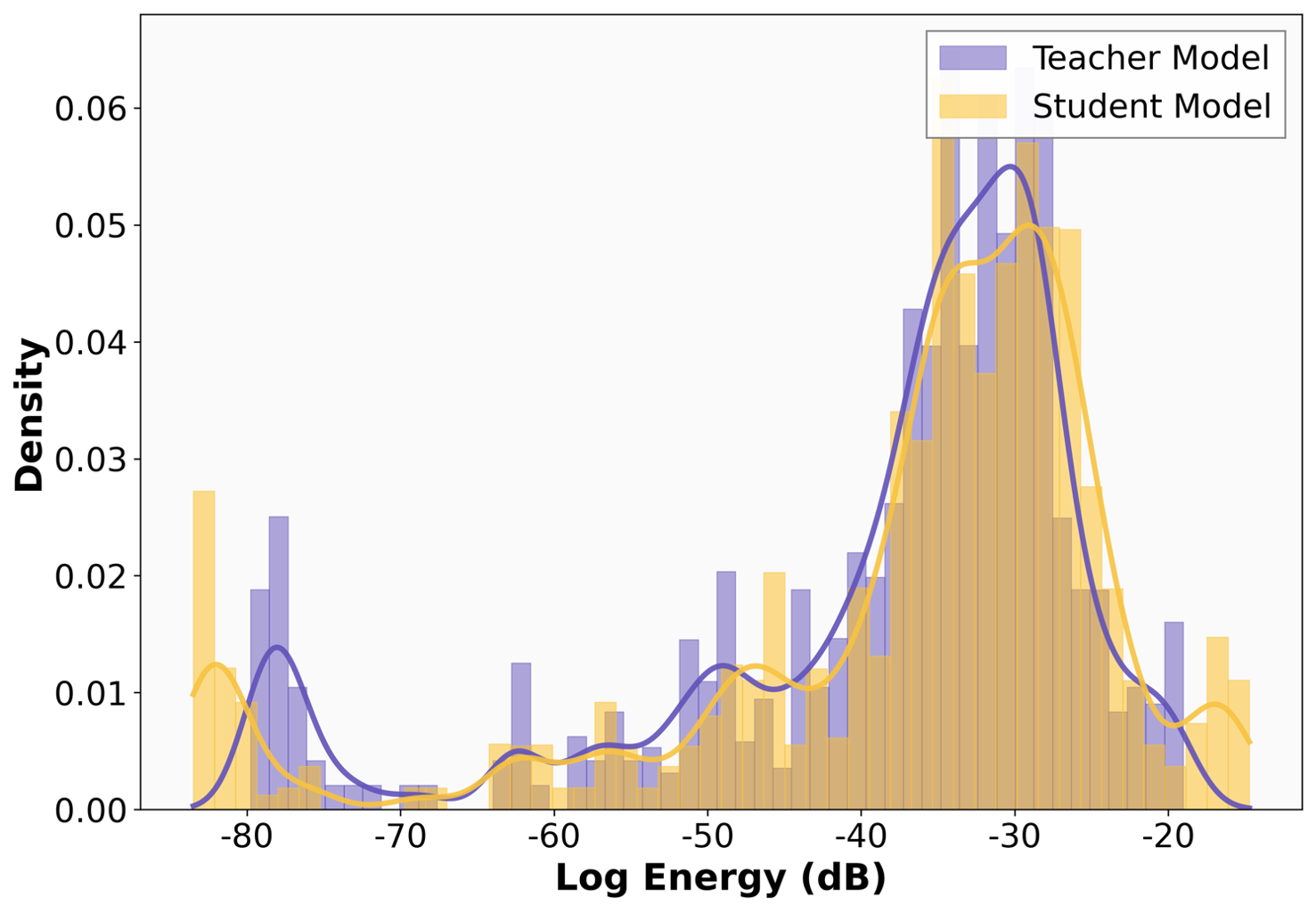}
\label{subfig:log_energy}
}\\[0.5em]  

\subfloat[Zero Crossing Rate (ZCR) distribution]{
\includegraphics[width=0.45\textwidth]{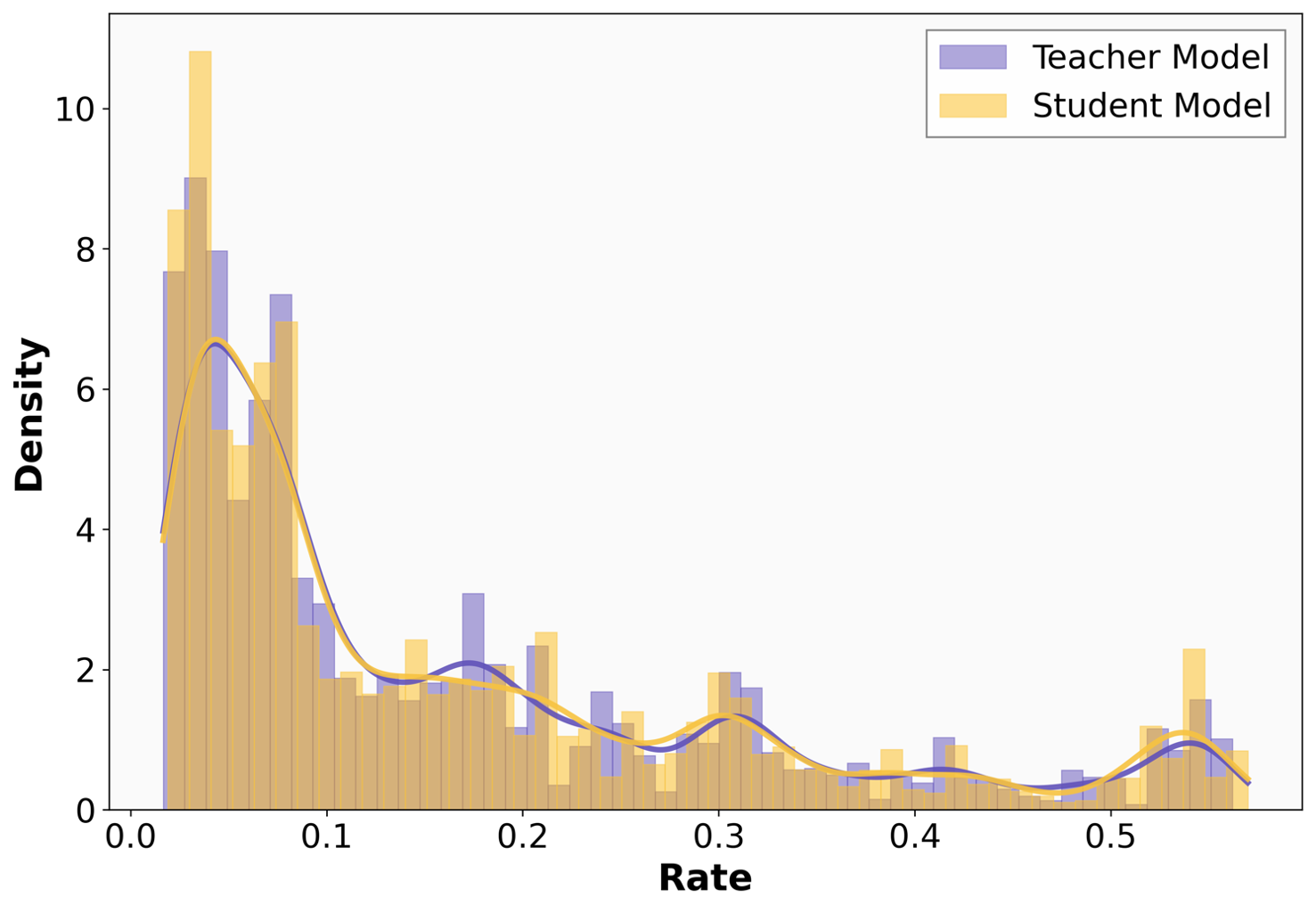}
\label{subfig:zcr}
}
\subfloat[Spectral Bandwidth distribution]{
\includegraphics[width=0.45\textwidth]{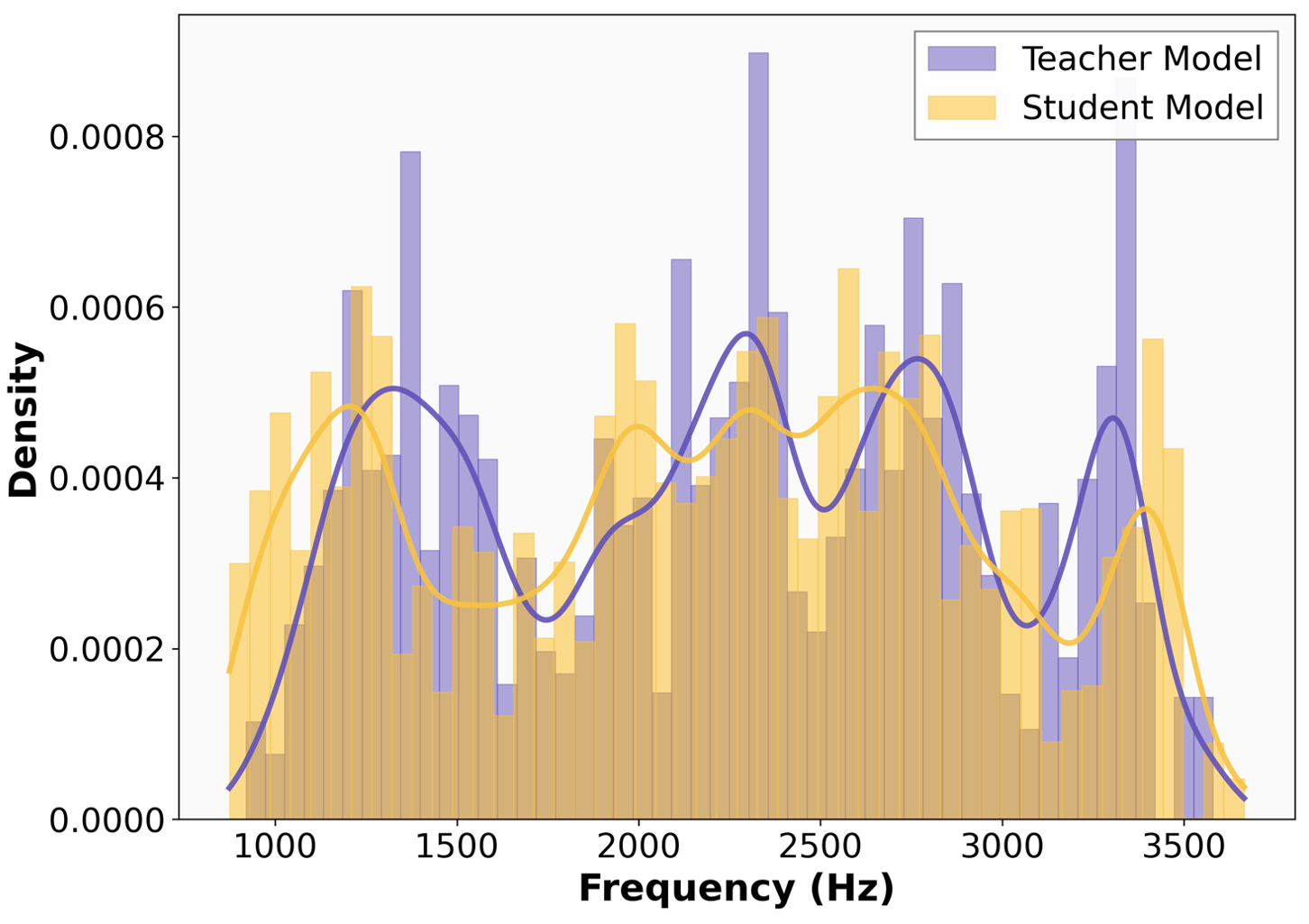}
\label{subfig:spectral_bandwidth}
}

\caption{Comparisons of key acoustic prosodic features between the teacher (10-step Flow Matching) and our student (1-step DSFlow) models. For each feature, 100 speech samples were generated under the same text and prompt condition, with frame-level values visualized via histograms and kernel density estimates. The student's distributions closely mirror the core mode of teacher while retaining natural variability—validating that our framework inherits the teacher's prosodic characteristics effectively.}
\label{fig:prosodic_features}
\end{figure*}

As illustrated in Figure~\ref{fig:prosodic_features}, the 1-step student model demonstrates precise distributional alignment with the 10-step teacher across all dimensions. This complements the correlation results in Table 3, further validating that our approach effectively preserves fine-grained prosodic attributes. Specifically, The student accurately replicates the teacher’s primary spectral peak around 1500–2000 Hz and the secondary modes at higher frequencies. This alignment is critical for inheriting the core timbre and vocal brightness of the teacher model. Regarding intensity, both models exhibit a distinct bimodal Log Energy distribution representing silence and speech intervals, with the student perfectly replicating the teacher's dynamic range from -80dB to -20dB. Furthermore, the skewed distribution of Zero Crossing Rate is matched with high fidelity, indicating correct distinction between voiced and unvoiced segments, while the Spectral Bandwidth analysis reveals that the student retains the complex, multi-modal envelope of the teacher's spectral characteristics rather than smoothing over these fine details. This consistent alignment confirms that our framework successfully inherits the teacher's prosodic regularity and acoustic richness while retaining natural variability.


\end{document}